\documentclass[11pt]{article}

\usepackage[fleqn]{amsmath}
\usepackage{amsthm,amssymb}

\usepackage{graphicx}
\usepackage[small]{caption}
\usepackage{subcaption}
\usepackage{epsfig}
\usepackage{color}
	
\usepackage{amsfonts}
\usepackage[utf8]{inputenc}
\usepackage{fullpage}
\usepackage{framed}

\usepackage{enumerate}
\usepackage{url}
\usepackage{hyperref}

\newtheorem{theorem}{Theorem}
\newtheorem{lemma}[theorem]{Lemma}

\newtheorem{corollary}[theorem]{Corollary}

\newcommand{\etal}{{et~al.}}
\newcommand{\ie}{{i.e.}}
\newcommand{\eg}{{e.g.}}

\newcommand{\RR}{\mathbb{R}} 
\newcommand{\eps}{\varepsilon}

\DeclareMathOperator{\polylog}{polylog}

\newcommand{\later}[1]{}
\newcommand{\old}[1]{}

\title{Arcs with increasing chords in $\RR^d$
   \footnote{Research partially supported
    by the Hungarian Science Foundation (NKFIH) grants  131529, 147544
   the ERC grants no. 882971, "GeoScape", and no. 101054936, "ERMid", and by the Erd\H os Center.}
}

\author{%
Adrian Dumitrescu\footnote{Algoresearch L.L.C., Milwaukee, WI 53217, USA,
  E-mail: \texttt{ad.dumitrescu@algoresearch.org}}
\and
Zsolt L\'angi\footnote{Bolyai Institute, University of Szeged, and
  Alfr\'ed R\'enyi Institute of Mathematics, Budapest, Hungary\@.
  Email: \texttt{zlangi@server.math.u-szeged.hu}}
}

\begin{document}

\maketitle

\begin{abstract}
  
  A curve $\gamma$ that connects $s$ and $t$ has the increasing chord property
  if $|bc| \leq |ad|$ whenever $a,b,c,d$ lie in that order on $\gamma$.
  For planar curves, the length of such a curve is known to be at most $2\pi/3 \cdot |st|$.
  Here we examine the question in higher dimensions and from the algorithmic standpoint
  and show the following:

  (I) The length of any $s-t$ curve with increasing chords in $\RR^d$ is at most
  $2 \cdot \left( e/2 \cdot (d+4) \right)^{d-1} \cdot |st|$ for every $d \geq 3$.
  This is the first bound in higher dimensions.

  (II)  Given a polygonal chain $P=(p_1, p_2, \dots, p_n)$ in $\RR^d$,
  where $d \geq 4$, $k =\lfloor d/2 \rfloor$,
  it can be tested whether it satisfies the  increasing chord property in
  $O\left(n^{2-1/(k+1)} \polylog(n) \right)$ expected time.
  This is the first subquadratic algorithm in higher dimensions.

\end{abstract}

\section{Introduction} \label{sec:intro}

A curve $\gamma$ has the \emph{increasing chord} property if $|bc| \leq |ad|$
whenever $a,b,c,d$ lie in that order on $\gamma$; see Fig.~\ref{fig:definition}\,(left).
In contrast, a curve $\gamma$ is said to be \emph{self-approaching} if $|bc| \leq |ac|$
whenever $a,b,c$ lie in that order on $\gamma$~\cite{IKL99}. As such, a path $\gamma$ has increasing
chords if and only if both $\gamma$ and $\gamma^R$ are self-approaching, where $\star^R$ denotes
path reversal.

Binmore~\cite{Bin71} asked whether there exists an absolute constant $c'$ such that $L \leq c'|st|$, where 
$\gamma$ is a plane curve with the increasing chord property from $s$ to $t$ and of length $L$.
Larman and McMullen~\cite{LM72} proved that one can take $c'=2 \sqrt3$, and twenty years later Rote~\cite{Rot94}
established that the value $c'=2\pi/3 = 2.094\ldots$ is the best possible. This bound is attained
by a curve consisting of two sides of a Reuleaux triangle; see Fig.~\ref{fig:definition}\,(right).

\begin{figure}[ht]
\begin{center}
\scalebox{0.7}{\includegraphics{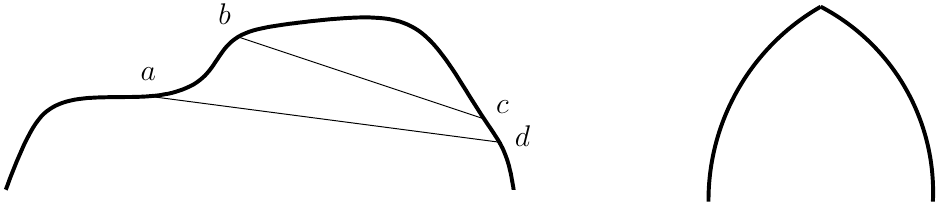}}
\caption{Left: a curve with the increasing chord property. Right: an arc consisting of two consecutive
  sides of a Reuleaux triangle.}
\label{fig:definition}
\end{center}
\end{figure}

The conjecture whether the longest curve with the increasing chord property in $\RR^d$, $d \geq 3$, 
is an arc consisting of $d$ consecutive sides of a Reuleaux simplex~\cite[G3]{CFG91} was refuted
by Rote~\cite{Rot94} already for $d=3$. His example was based on the observation that if
$S$ is a Reuleaux unit tetrahedron\footnote{A $d$-dimensional Reuleaux unit simplex
  is obtained as the intersection of $d+1$ unit balls, centered at the vertices of a regular $d$-simplex
  of unit edge length. A $3$-dimensional Reuleaux unit simplex is also called a Reuleaux unit tehrahedron.}
in $\RR^3$, then the midpoints of two disjoint edges of $S$ are at distance $\sqrt{3}-\sqrt{1/2} > 1$.
In addition, he proposed a slightly modified curve with the increasing chord property of length
about $3.087$, which, to the best of our knowledge, is the current record.
While it is not so easy to come up with a construction that works for any dimension $d \geq 3$,
an elementary computation shows that if $S$ is a Reuleaux unit simplex
in $\RR^d$, with $d \geq 3$,
then the midpoints of two disjoint edges of $S$ are at distance
$\frac{1}{d-1} \left( \sqrt{d(d+1)}-\sqrt{2} \right) > 1$, leading to a counterexample to the above conjecture
for any fixed $d \geq 3$.

In this paper we obtain an explicit upper bound on the length of a curve with the increasing
chord property in $\RR^d$ and give a  subquadratic-time algorithm for determining whether
a given polygonal chain in $\RR^d$ satisfies the increasing chord property.

\paragraph{Definitions and notations.} 
Let $\gamma : [0,1] \to \RR^d$ be a continuous curve (path) between $\gamma(0)=s$ and $\gamma(1)=t$;
we sometimes refer to it as an $s-t$ curve.
All curves discussed in this paper are assumed to be \emph{piecewise smooth}. 

Following~\cite{HK22,IKL99}, a \emph{normal} to $\gamma$ at a point $p \in \gamma$  is any hyperplane
that is included in the double wedge between the hyperplanes orthogonal to the one-sided tangents
of the two smooth pieces of $\gamma$ meeting at $p$; note that there is a unique normal at a  point
$p$ of $\gamma$ iff $\gamma$ is smooth at $p$. If $p$ is not a smooth point, we call the two hyperplanes
orthogonal to the two one-sided tangents at $p$ \emph{extremal normals}.

The length of a segment $ab$ or a vector $v$ is denoted by $|ab|$ or $|v|$, respectively. For brevity,
we use the same symbol $|\gamma|$ to denote the arc-length of the (rectifiable) curve $\gamma$. We identify
points and their position vectors; in particular, we denote the vector directed from $a$ to $b$ by $b-a$.

For any two points $a,b \in \gamma$, consider the \emph{detour} between the two points, 
namely the ratio between the length of the subpath $|\gamma(a,b)|$ and the Euclidean distance $|ab|$.
The following parameters are studied, \eg, in~\cite{CDMT21,DEK+07,NS00}.

The \emph{geometric dilation} of $\gamma$ is
\[ \delta(\gamma)=\sup_{a,b \in \gamma} \frac{|\gamma(a,b)|}{|ab|}. \]

The \emph{stretch factor} $\delta_{s,t}(\gamma)$ of $\gamma$ is defined as the detour between
the two endpoints $s$ and $t$, namely:
\[ \delta_{s,t}(\gamma) = \frac{|\gamma(s,t)|}{|st|}. \]

The $\beta$-skeleton of a set of points in $\RR^d$ is a geometric graph defined on this set,
in which two points $a,b$ are connected by an edge if no point $c$ of the set forms an angle $\angle{acb}$ with $ab$ greater than
$\arcsin(1/\beta)$ (if $\beta>1$), or $\pi - \arcsin \beta$ (if $\beta<1$)~\cite{Ep02}.

Polylogarithmic functions are defined as functions that grow at most polynomially with respect to the
logarithm of their input, often denoted as $O(\log^k n)$, for some fixed $k$, or simply $\polylog(n)$.

\paragraph{Our results.} 

In Section~\ref{sec:d-space} we prove:

\begin{theorem}\label{thm:d-space}
  Let $d \geq 3$. Let $\gamma: [0,1] \to \RR^d$ be a curve with $|\gamma(0)\gamma(1)| = 1$
  and satisfying the increasing chord property. Then the following inequalities hold.
\begin{itemize}
\item[(i)] For any $0 < \alpha < \frac{\pi}{2}$, the length of $\gamma$ is at most
\begin{equation} \label{eq:highdUB1}
|\gamma| \leq (1+\cos \alpha) \left( \frac{2-\sin \alpha}{1-\sin \alpha} \right)^{d-1}
\cdot \frac{1}{\left( \sin \alpha \right)^{\binom{d}{2}}}.
\end{equation}
\item[(ii)] In particular, the length of $\gamma$ is at most
\begin{equation} \label{eq:highdUB2}
|\gamma| \leq 2 \left( \frac{e}{2} \cdot (d+4) \right)^{d-1}.
\end{equation}
\end{itemize}
\end{theorem}

We note that the results in Theorem~\ref{thm:d-space} are valid for any (not necessarily piecewise smooth)
curve satisfying the increasing chord property. Furthermore, for any fixed value of $d \geq 3$, by elementary calculus
it is possible to minimize the expression on the right-hand side of~\eqref{eq:highdUB1}, yielding an explicit
upper bound on the stretch factor of $\gamma$.

Turning to the algorithmic problem of determining whether a given polygonal chain
in $\RR^d$ satisfies the increasing chord property, in Section~\ref{sec:alg} we prove:

\begin{theorem} \label{thm:alg}  
  Given a polygonal chain $P=(p_1, p_2, \dots, p_n)$ in $\RR^d$, where $d \geq 4$, and $k =\lfloor d/2 \rfloor$,
  it can be determined by a randomized algorithm running in $O\left(n^{2-1/(k+1)} \polylog(n) \right)$ expected time,
  whether $P$ satisfies the increasing chord property. 
\end{theorem}

Finally, in Section~\ref{sec:remarks} we list two open problems suggested by this work.

\paragraph{Related work.} 
Alamdari~\etal~\cite{ACG+13} showed that testing whether a given polygonal arc on $n$ points in $\RR^d$
is self-approaching can be done in $O(n)$ time for $d=2$ and in $O(n \log^2{n}/\log \log{n})$ time for $d=3$. 
One motivation for studying increasing chord curves and self-approaching curves comes from the fact that
planar curves in both classes have a small geometric dilation, \ie, at most $2.094$ and $5.3332$, respectively;
see~\cite{AKW04,IKL99}.
In contrast, Eppstein~\cite{Ep02} showed that there are points sets in $\RR^2$ and values $\beta$ for which
the $\beta$-skeleton of the set is a polygonal chain with an arbitrarily large stretch factor.

The problem of computing the stretch factors of paths, cycles, and other structures was studied by
Agarwal~\etal~\cite{AKK+08}, Chen~\etal~\cite{CDMT21},
Ebbers-Baumann~\etal~\cite{EKLL04}, Klein~\etal~\cite{KKNS09}, and Narasimhan and Smid~\cite{NS00,NS07},
among others.

\section{Preliminaries}\label{sec:prelim}

Let $\gamma$ be a piecewise smooth $s-t$ curve.
The following lemma, appearing in the work of Hagedoorn and Kostitsyna~\cite{HK22},
refers to curves in the plane, but the same proof also works in higher dimensions.

\begin{lemma} \label{lem:HK}  {\rm \cite{HK22}}
  An $s-t$ curve $\gamma$ in $\RR^d$ has increasing chords if and only if any normal to $\gamma$ at any point
  $p \in \gamma$ does not intersect the open subcurves $\gamma(s,p)$ and $\gamma(p,t)$.
\end{lemma}

This property can be also formulated in terms of \emph{signed} halfspaces~\cite{BKL20,HK22}.
A \emph{positive halfspace} $h_p^+$ of $\gamma$ at point $p \in \gamma$ is a closed halfspace
bounded by a normal of $\gamma$ at $p$, which contains a neighborhood of $p$ in $\gamma(p,b)$.
Similarly, a \emph{negative halfspace} $h_p^-$ of $\gamma$ at point $p \in \gamma$ is a closed halfspace
bounded by a normal of $\gamma$ at $p$, which contains a neighborhood of $p$ in $\gamma(a,p)$.
If $p \in \gamma$ is a nonsmooth point, we say that the negative halfspace bounded by the normal perpendicular
to the left-sided tangent at $p$, and the positive halfspace bounded by the normal perpendicular
to the right-sided tangent at $p$ are \emph{extremal}. The reformulation is as follows: 

\begin{corollary} \label{cor:HK}  {\rm \cite{HK22}}
  An $s-t$ curve $\gamma$ in $\RR^d$ has increasing chords if and only if for any point $p \in \gamma$
  and positive and negative halfspaces $h_p^+$ and $h_p^-$ of $\gamma$ at point $p \in \gamma$,
  the subcurve $\gamma(s,p)$ is contained in the negative halfspace $h_p^-$ and
  the subcurve $\gamma(p,t)$ is contained in the positive halfspace $h_p^+$. 
\end{corollary}

\paragraph{Example.}
An \emph{ascending staircase} is a polygonal curve consisting of horizontal and vertical segments
with the $x$- and $y$- positive axis orientations, see Fig.~\ref{fig:examples}\,(left).
By Corollary~\ref{cor:HK}, such a curve has the increasing chord property; its length is
at most $\sqrt2 |st|$; this bound can be attained.  Two other examples appear in the same figure.

\begin{figure}[ht]
\begin{center}
\scalebox{0.7}{\includegraphics{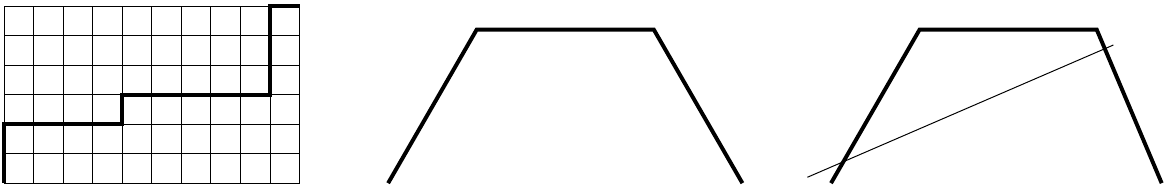}}
\caption{Left and center: two curves with increasing chords (an orthogonal staircase and 
 the upper hull of a regular hexagon).
  Right: a curve that does not satisfy the increasing chord property and a witness line.}
\label{fig:examples}
\end{center}
\end{figure}

\medskip
The above lemma and corollary allow one to obtain a subquadratic time algorithm for testing the
increasing chord property.

\section{Higher dimensional Euclidean upper bound} \label{sec:d-space}

Recall that a curve $\gamma:[0,1] \to \RR^d$ is monotone in direction $q$ if
$\tau \mapsto \langle \gamma(\tau), q \rangle $ is a monotone function;
here, by a \emph{direction} we mean a nonzero vector in $\RR^d$,
and $\langle x,y \rangle$ denotes the \emph{inner product} of $x$ and $y$. The next lemma is a
straightforward generalization of \cite[Lemma 1]{Rot94}. 

\begin{lemma}\label{lem:monotonicity}
  Let $\gamma:[0,1] \to \RR^d$ satisfy the increasing chord property. Then $\gamma$ is monotone in the direction
  $\gamma(1)-\gamma(0)$.
\end{lemma}

\begin{proof}
  This property readily follows from the observation that for any $0 \leq \tau \leq 1$, the points of $\gamma([0,\tau])$
  belong to the closed ball of radius $|\gamma(0)\gamma(\tau)|$ centered at $\gamma(0)$, and do not belong to the open ball
  of radius $|\gamma(\tau)\gamma(1)|$ centered at $\gamma(1)$. 
\end{proof}

In the following Lemmas~\ref{lem:Rote1}-\ref{lem:induction} we assume that $d \geq 2$.
Lemma~\ref{lem:Rote1} below was stated and proved by Rote \cite[Lemma 4]{Rot94}.

\begin{lemma}\label{lem:Rote1}
  Suppose that a curve $\gamma:[0,1] \to \RR^d$ is monotone in the $d$ linearly independent directions
  $q_1, \ldots, q_d$. Let $s=\gamma(0)$ and $t=\gamma(1)$. Then the curve is contained in the parallelotope 
\[ P= \{ x \in \RR^d : \langle s, q_i \rangle \leq \langle x, q_i \rangle \leq \langle t, q_i \rangle \hbox{ for }
i=1,2,\ldots, d \}.   \]
Furthermore, its length is bounded by the sum of the lengths of $d$ successive edges of $P$ leading
from $s$ to $t$.
\end{lemma}

We continue with a generalization of \cite[Lemma 5]{Rot94} for dimension $d$. The obtained bounds are relevant
in bounding from above the lengths of $d$ successive edges of $P$. We note that our argument essentially differs
from that of~\cite[Lemma 5]{Rot94}.

\begin{lemma}\label{lem:new}
  Let $q_1, \ldots, q_d \in \RR^d$ be unit vectors such that the angle between $q_2$ and $q_1$ is
  $0 < \alpha < \frac{\pi}{2}$, the angle between $q_3$ and the plane of $q_1, q_2$ is $\alpha$, and so on,
  \ie, the angle between $q_d$ and the hyperplane of the other vectors is $\alpha$.
  Let $Q$ be the $d \times d$ matrix whose $i$th row is $q_i$ for
$i=1,2,\ldots, d$. Let $r_i$ denote the $i$th column of $Q^{-1}$. Then 
\[
| r_1 | = \frac{1}{\sin^{d-1} \alpha}, \hbox{ and}
\]
\[
|r_i| \leq \frac{1}{\sin^{d-i+1} \alpha}  \hbox{ for all } 1 < i \leq d.
\]
\end{lemma}

\begin{proof}
Without loss of generality, we may assume that $q_1$ coincides with the first basis vector, $q_2$ lies in the
$(x_1,x_2)$-plane, and so on, i.e. $q_i$ lies in the linear subspace spanned by the first $i$ basis vectors for
$i=1,2,\ldots, d$. Then we can write $Q$ as 
\[
Q = \left[ \begin{array}{rrrr}
1 & 0 & \ldots & 0\\
\cos \alpha & \sin \alpha & \ldots & 0\\
\cos \zeta_{31} \cos \alpha & \sin \zeta_{31} \cos \alpha & \ldots & 0\\
\vdots & \vdots & & \vdots \\
\cos \zeta_{d(d-2)} \ldots \cos \zeta_{d1} \cos \alpha & \sin \zeta_{d(d-2)} \cos \zeta_{d(d-3)}  
\ldots \cos \zeta_{d1} \cos \alpha & \cdots & \sin \alpha
\end{array}  
\right]
\]
for some suitable angles $\zeta_{uv}$.

Let $Q_{uv}$ denote the matrix obtained from $Q$ by removing its $u$th row and $v$th column.
If $r_{uv}$ denotes the entry of $Q^{-1}$ in the $u$th row and $v$th column,
then $r_{uv} = \frac{ (-1)^{u+v} \det (Q_{vu})}{\det (Q)}$. We estimate $| \det(Q_{uv})|$. 

First, recall that the inverse of a lower triangular matrix is lower triangular, and thus,
if $v < u$, then $\det(Q_{uv}) = 0$.
Recall also that the determinant of a lower triangular matrix is the product of its diagonal elements.
Thus, we have $\det(Q_{uu}) = \sin^{d-2} \alpha$ for all $u \geq 2$, and $\det(Q_{11}) = \sin^{d-1} \alpha$.
From now on we assume that $u < v$.

Next, we consider the case that $1 < u < v-1 < d-1$. Then $Q_{uv}$ is a block lower triangular matrix
consisting of three square blocks $B,C,D$ of sizes $u-1$, $v-u$ and $d-v$, respectively. Specifically,
if $O_{k \times l}$ denotes the zero matrix of size $k \times l$, then $Q_{uv}$ can be written as 
\[
Q_{uv} = \left[ \begin{array}{rrr}
B & O_{(u-1) \times (v-u)} & O_{(u-1) \times (d-v)}  \\
E & C & O_{(v-u) \times (d-v)}\\
F & G & D
\end{array}  
\right]
\]
for some rectangular matrices $E, F, G$.

Thus, we have $\det Q_{uv} = (-1)^{u+v} \det(B) \det(C) \det(D)$. Here, $B$ and $D$ are lower triangular matrices, and
their determinants are $\sin^{u-2} \alpha$ and $\sin^{d-v} \alpha$, respectively. The second block $C$ is not lower
triangular, since the elements right above the main diagonal are equal to $\sin (\alpha)$. On the other hand, to
estimate $\det(C)$ we may use the geometric interpretation of the determinant of a matrix: it is equal to the signed
volume of the parallelotope induced by its row vectors. Thus, $|\det (C)|$ is less than or equal to the product of the
lengths of its row vectors. 

The first row of $C$ is $[ \begin{array}{ccccc} \sin \zeta_{(u+1)1} \cos \alpha & \sin \alpha & 0 & \ldots &
    0 \end{array} ]$, and its length is at most $1$, with equality if $|\sin \zeta_{(u+1)1}| = 1$.
Similarly, the second row of $C$ is 
\[
  [ \begin{array}{cccccc} \sin \zeta_{(u+2)2} \cos \zeta_{(u+2)1} \cos \alpha &
      \sin \zeta_{(u+2)1} \cos \alpha & \sin \alpha & 0 & \ldots & 0 \end{array} ], 
\]
and its length is at most $1$. Using the same consideration, we obtain that the length of every row of $C$,
but the last one, is at most one, and the length of the last row of $C$ is at most $\cos \alpha$.
This implies that if $1 < u < v-1 < d-1$, then $| \det(Q_{uv})| \leq \cos \alpha \sin^{d+u-v-2} \alpha$. 

We are left with the cases that $u > 1$, and $u=v-1$ or $v=d$, or that $u=1$ and $v > 1$. A slight modification
of the previous argument shows that in the first case $|\det(Q_{uv})| \leq \cos \alpha \sin^{d+u-v-2} \alpha$,
and in the second case $|\det(Q_{1v})| \leq \cos \alpha \sin^{d-v} \alpha$. 
Summing up,
\begin{itemize}
\item[(i)] If $u > v$, then $\det(Q_{uv})=0$;
\item[(ii)] $\det(Q_{11}) = \sin^{d-1} \alpha$, and if $u > 1$, $\det(Q_{uu}) = \sin^{d-2} \alpha$;
\item[(iii)] If $1 = u < v$, then $|\det(Q_{uv})|=|\det(Q_{1v})| \leq \cos \alpha \sin^{d-v} \alpha$;
\item[(iv)] If $1 < u < v$, then $|\det(Q_{uv})| \leq \cos \alpha \sin^{d+u-v-2} \alpha$.
\end{itemize}

Now we estimate the lengths of the $r_i$. Recall that $r_i$ is the $i$th column of $Q^{-1}$. Thus, by the Pythagorean
Theorem, the above estimates and the summation formula for the elements of a geometric sequence, we obtain that if $1 <
i < d$, then
\begin{align*}
|r_i| &= \sqrt{\sum_{j=1}^d r_{ji}^2} = \frac{1}{\sin^{d-1} \alpha} \sqrt{\sum_{j=1}^d \left(\det(Q_{ij})\right)^2} \\
&\leq \frac{1}{\sin^{d-1} \alpha} \sqrt{ \sin^{2d-4} \alpha + \cos^2 \alpha \left( \sin^{2d-6} \alpha + \ldots
  \sin^{2i-4} \alpha \right) } = \frac{1}{\sin^{d-i+1}\alpha}. 
\end{align*}
Similarly, we obtain that $|r_d| \leq \frac{1}{\sin \alpha}$ and $|r_1| \leq \frac{1}{\sin^{d-1} \alpha}$.
\end{proof}

\begin{lemma}\label{lem:Rote2}
Let $\gamma:[0,1] \to \RR^d$ be a curve that is monotone in the $d$ linearly independent directions $q_1, \ldots, q_d$. Let $s=\gamma(0)$
and $t=\gamma(1)$. Assume that the angle between $q_2$ and $q_1$ is $0 < \alpha < \frac{\pi}{2}$, the angle between $q_3$ and
the plane of $q_1, q_2$ is $\alpha$, and so on i.e. the angle between $q_d$ and the hyperplane of the other vectors is
$\alpha$. Assume also that $t-s = q_d$. Then the length of $\gamma$ is at most
\[
|\gamma| \leq |st| \left( \frac{1}{\sin \alpha} + \frac{\cos \alpha}{\sin^{d-1} \alpha} +  \frac{\cos \alpha \left( 1- \sin^{d-2}
  \alpha \right)}{\sin^{d-1} \alpha \left( 1 - \sin \alpha \right)} \right). 
\]
\end{lemma}

\begin{proof}
Without loss of generality, we may assume that all $q_i$ are unit vectors.
We use the result of Lemma~\ref{lem:Rote1} and the notation of Lemma~\ref{lem:new}. By the definition of the inverse of
a matrix, for all $i,j$ we have $\langle q_i, r_j \rangle = \delta_{ij}$, where $\delta_{ij}$ denotes the Kronecker
delta. Observe that for every value of $i$, there is an edge class of $P$ perpendicular to all $q_j$ but $q_i$. Let
$E_i$ denote an edge from this class starting at $s$. Then $E_i$ is parallel to $r_i$. 

We may assume for simplicity that $s$ is the origin. Observe that the orthogonal projection of $E_i$ onto the line of $q_i$
coincides with that of $st$.
Thus, if $E_i$ is the segment with endpoints $s=o$ and $\lambda r_i$, then
$\langle q_d, q_i \rangle = \lambda \langle r_i, q_i \rangle = \lambda$,
implying that $|E_i| = |\langle q_d, q_i \rangle| \cdot |r_i|$.
On the other hand, our conditions for the $q_i$ yield that the angle between $q_i$ and
$q_j$ is at least $\alpha$ and at most $\pi - \alpha$ for any $i \neq j$, implying that
$|\langle q_d, q_i \rangle | \leq \cos \alpha$ if $i \neq d$. Hence, 
\[
\sum_{i=1}^d |E_i| \leq |r_d| + \cos \alpha \sum_{i=1}^{d-1} |r_i| \leq \frac{1}{\sin \alpha} + \frac{\cos \alpha
}{\sin^{d-1} \alpha} + \sum_{i=2}^{d-1} \frac{\cos \alpha}{\sin^i \alpha}, 
\]
from which the assertion follows.
\end{proof}

Based on Lemma~\ref{lem:Rote2}, we set $C_1(\alpha)=1$, and  
\begin{equation} \label{eq:C_d}
  C_d(\alpha) = \frac{1}{\sin \alpha} + \frac{\cos \alpha}{\sin^{d-1} \alpha} +
  \frac{\cos \alpha \left( 1- \sin^{d-2} \alpha \right)}{\sin^{d-1} \alpha \left( 1 - \sin \alpha \right)}, \text{ for }
  d \geq 2.
\end{equation}
In particular, we have  $C_2(\alpha)= \frac{1+\cos \alpha}{\sin \alpha}$.

Define $F_i(\alpha)$ inductively for any $1 \leq i \leq d$ in the following way:
\begin{equation} \label{eq:F_d}
F_{d}(\alpha)= C_d(\alpha), \text{ and }
F_i(\alpha) = \left( 1+\frac{C_i(\alpha)}{\cos \alpha} \right) F_{i+1}(\alpha), \text{ for } 1 \leq i < d.
\end{equation}
We note that, as an elementary computation shows,
for any $0 <\alpha < \frac{\pi}{2}$ and $1 \leq i \leq d$, we have $F_i(\alpha) >1$.   

\begin{lemma}\label{lem:induction}
Let $\gamma:[0,1] \to \RR^d$ be a curve that satisfies the increasing chord property, and is monotone in $i \geq 1$ linearly
independent directions $q_1, \ldots, q_i$. Let $s=\gamma(0)$ and $t=\gamma(1)$ such that $\gamma(1)-\gamma(0)=q_i$ and $|q_i|=1$.
Assume that the angle between $q_2$ and $q_1$ is $0 < \alpha < \frac{\pi}{2}$, the angle between $q_3$ and the plane of
$q_1,q_2$ is $\alpha$, and so on, i.e. the angle between $q_i$ and the linear subspace spanned by $q_1, q_2, \ldots, q_{i-1}$
is $\alpha$. Then the length of $\gamma$ is at most $|\gamma| \leq F_i(\alpha)$. 
\end{lemma}

\begin{proof}
We prove the statement by induction on $i$ from $i=d$ down to $i=1$. The base case $i=d$ of the lemma is proved in
Lemma~\ref{lem:Rote2}. Let $1 \leq i \leq d-1$, and assume that the statement holds for $i+1$ linearly independent
directions. We need to show that it holds for $i$ linearly independent directions. 
 
Let $L$ denote the $i$-dimensional linear subspace spanned by the $q_j$s. Consider a subdivision
$0 = \tau_0 < \tau_1 < \tau_2 < \ldots < \tau_n = 1$ of the interval $[0,1]$ with the property that the length of
the polygonal chain $\left(\gamma(\tau_0), \gamma(\tau_1), \ldots, \gamma(\tau_n) \right)$
approximates the length of $\gamma$ within the error $\eps > 0$ for an arbitrary fixed value of $\eps$.
We call a segment $\gamma(\tau_j)\gamma(\tau_{j+1})$ \emph{flat} if its angle with $L$ is at most $\alpha$,
and \emph{steep} if this angle is greater than $\alpha$. 

Now, for every steep segment $\gamma(\tau_j)\gamma(\tau_{j+1})$ we define a \emph{covering interval} $[u,v]$ with
$u \leq \tau_j$ and $v \geq \tau_{j+1}$ such that the angle between $\gamma(u)\gamma(v)$ and $L$ is exactly $\alpha$;
  we do it in such a way that the covering intervals do not overlap too much. 
To do it, first set $j=0$. If $\gamma(\tau_j)\gamma(\tau_{j+1})$ is flat, we proceed to the next segment.

Consider the case that $\gamma(\tau_j)\gamma(\tau_{j+1})$ is steep. Then we take the last value $\bar{\tau} \leq 1$ such that the angle
between $\gamma(\tau_j)\gamma(\bar{\tau})$ and $L$ is $\alpha$. We add $\bar{\tau}$ to the set of division points of $[0,1]$, and if
$\tau_k \leq \bar{\tau} < \tau_{k+1}$, we proceed to the segment $\gamma(\bar{\tau})\gamma(\tau_{k+1})$. Note that if the angle between
$\gamma(\tau_j)\gamma(1)$ and $L$ is at most $\alpha$, such a point exists by continuity. Assume that no such point exists,
implying that the angle between $\gamma(\tau_j)\gamma(1)$ and $L$ is greater than $\alpha$. Then we stop the procedure, and find
the first point $0 \leq \bar{\tau} < \tau_j$ such that the angle between $\gamma(\bar{\tau})\gamma(1)$ is equal to $\alpha$. Since
$\gamma(0)\gamma(1)$ is parallel to $L$ this value $\bar{\tau}$ exists, and satisfies $\bar{\tau} > 0$. In this case we call this
interval $[\bar{\tau},1]$ a \emph{special covering interval}. Note that, according to our construction, apart from the
last, special covering interval if it exists, all covering intervals are pairwise nonoverlapping, and for any steep
segment $\gamma(\tau_j)\gamma(\tau_{j+1})$, there is a covering interval containing $[\tau_j,\tau_{j+1}]$. In the following,
if $[u,v]$ is a covering interval, we define the segment $\gamma(u)\gamma(v)$ a \emph{covering segment}.

Let $\gamma_L$ denote the orthogonal projection of $\gamma$ onto $L$. Let $a$ denote the total length of the projections of all
flat segments whose parameter range is not covered by covering intervals, $b$ denote the total length of the projections
of all covering segments but the last, special one, if it exists, and let $c$ denote the length of the projection of the
special covering segment, if it exists; otherwise set $c=0$.

By Lemma~\ref{lem:Rote2}, the length of $\gamma_L$ is at most $C_i(\alpha)$.
Thus, $a+b \leq C_i(\alpha)$, as the corresponding segments are mutually
nonoverlapping. Furthermore, by the increasing chord property, the length of the special covering segment, if it
exists, is at most $1$, implying that $c \leq \cos \alpha$. 

For any flat segment, if its projection has length $x$, then the length of the segment is at most
$\frac{x}{\cos \alpha}$, implying that the total length of the considered flat segments is at most $\frac{a}{\cos \alpha}$.
Now, consider a covering segment $\gamma(u)\gamma(v)$. The curve $\gamma([u,v])$ satisfies the increasing chord property
and is monotone in the direction of $q_{i+1}=\gamma(v)-\gamma(u)$, and hence, by the induction hypothesis,
its length is at most $|f(\tau_u)f(\tau_v)| \cdot F_{i+1}(\alpha)$. Since the length of the projection of this covering segment
is $\cos \alpha |\gamma(u)\gamma(v)|$, it follows that the length of the part of $\gamma$ covered by the covering intervals
is at most $(b+c)\cdot \frac{F_{i+1}(\alpha)}{\cos\alpha}$.
Thus, the length of $\gamma$ is at most 
\[
\frac{a}{\cos \alpha} + (b+c) \frac{F_{i+1}(\alpha)}{\cos \alpha} +  \eps,
\]
where $0 \leq a,b,c$, $a+b \leq C_i(\alpha)$ and $c \leq \cos \alpha$. Since $F_{i+1}(\alpha) > 1$, the above expression
is maximized if $a=0$, $b=C_i(\alpha)$ and $c=\cos \alpha$. As the obtained inequality holds for any fixed value
$\eps > 0$, it follows that the length of $\gamma$ is at most 
\begin{equation*}  
|\gamma| \leq F_i(\alpha) = \left( 1 + \frac{C_i(\alpha)}{\cos \alpha} \right) F_{i+1}(\alpha). \qedhere 
\end{equation*}  
\end{proof}

\begin{lemma}\label{lem:estimate}
For any $2 \leq i \leq d$ and $0 < \alpha < \frac{\pi}{2}$, we have
\[
C_d(\alpha) \leq \frac{\cos \alpha (2-\sin \alpha)}{(1-\sin \alpha )\sin^{d-1} \alpha}; \text{~~and~~}
1 + \frac{C_i(\alpha)}{\cos \alpha} \leq \frac{2-\sin \alpha}{(1-\sin \alpha) \sin^{i-1} \alpha}. 
\]
Consequently, we have
\[
F_1 \left( \alpha \right) \leq (1+\cos \alpha) \cdot \left( \frac{2-\sin \alpha}{1-\sin \alpha} \right)^{d-1}
\cdot \frac{1}{\left( \sin \alpha \right)^{ \binom{d}{2} } }. 
\]
\end{lemma}

\begin{proof}
Recall that $C_1(\alpha)=1$, and by~\eqref{eq:C_d}, we have 
\begin{align*}
C_d( \alpha) &=  \frac{1}{\sin \alpha} + \frac{\cos \alpha}{\sin^{d-1} \alpha} +
\frac{( 1- \sin^{d-2} \alpha ) \cos \alpha }{( 1 - \sin \alpha) \sin^{d-1} \alpha } \\
&= \frac{(1-\sin \alpha) \sin^{d-2} \alpha + (1- \sin \alpha) \cos \alpha + (1 - \sin^{d-2} \alpha) \cos \alpha}
     { ( 1 - \sin \alpha ) \sin^{d-1} \alpha } \\
     &= \frac{ (2 -\sin \alpha) \cos \alpha} {( 1 - \sin \alpha ) \sin^{d-1} \alpha} +
     \frac{ (1 -\sin \alpha - \cos \alpha) \sin^{d-2} \alpha} {( 1 - \sin \alpha ) \sin^{d-1} \alpha} \\
     &\leq \frac{ (2 -\sin \alpha) \cos \alpha} {( 1 - \sin \alpha ) \sin^{d-1} \alpha},
\end{align*}
where we used the inequality $\sin \alpha + \cos \alpha > 1$ for any $0 < \alpha < \frac{\pi}{2}$ in the last step.

Similarly, by~\eqref{eq:C_d}, for any $i \geq 2$, 
\begin{align*}
  1+\frac{C_i(\alpha)}{\cos \alpha} &= \frac{ (2 -\sin \alpha)} {( 1 - \sin \alpha ) \sin^{i-1} \alpha}
  + \left( \frac{1 -\sin \alpha - \cos \alpha}{( 1 - \sin \alpha ) \sin \alpha \cos \alpha} +1 \right) \\
  &= \frac{ (2 -\sin \alpha)} {( 1 - \sin \alpha ) \sin^{i-1} \alpha}
  + \frac{(1-\cos \alpha ) (\sin^2 \alpha - \cos \alpha - \sin \alpha)}{( 1 - \sin \alpha ) \sin \alpha \cos \alpha}  \\
&< \frac{ (2 -\sin \alpha)} {( 1 - \sin \alpha ) \sin^{i-1} \alpha}.
\end{align*}

We deduce the last formula inductively:
\begin{align*}
  F_1(\alpha) &= C_d(\alpha) \prod_{i=1}^{d-1} \left( 1+ \frac{C_i(\alpha)}{\cos \alpha}\right)
  = \frac{1+\cos \alpha}{\cos \alpha } \cdot C_d(\alpha) \cdot \prod_{i=2}^{d-1} \left( 1+ \frac{C_i(\alpha)}{\cos \alpha}\right)\\
  &\leq  \frac{1+\cos \alpha}{\cos \alpha } \cdot \frac{ (2 -\sin \alpha) \cos \alpha} {( 1 - \sin \alpha ) \sin^{d-1} \alpha}
  \cdot \prod_{i=2}^{d-1} \frac{ (2 -\sin \alpha)} {( 1 - \sin \alpha ) \sin^{i-1} \alpha} \\
  &= (1+\cos \alpha) \cdot \left( \frac{2-\sin \alpha}{1-\sin \alpha} \right)^{d-1} \cdot
  \frac{1}{\left( \sin \alpha \right)^{ \binom{d}{2} }}.   \qedhere 
\end{align*}
\end{proof}

Now we prove Theorem~\ref{thm:d-space}. Note that the statement in (i) readily follows from
Lemmas~\ref{lem:monotonicity}, \ref{lem:induction} and \ref{lem:estimate}.

To prove (ii), assume that $d \geq 3$, and let $\alpha = \arcsin \frac{d}{d+2}$, \ie, $\sin \alpha = \frac{d}{d+2}$.
Straightforward calculations yield that
\[ \left( \frac{2-\sin \alpha}{1-\sin \alpha} \right)^{d-1} = \left( \frac{d+4}{2} \right)^{d-1}, \text{ and }
\frac{1}{\sin ^{ \binom{d}{2} } \alpha} = \left( 1 + \frac{2}{d} \right)^{ \binom{d}{2} } 
= \left(1 +\frac{2}{d} \right)^{\frac{d}{2} \cdot (d-1)} \leq e^{d-1}. \]
Since $1+\cos \alpha < 2$, we have
\begin{equation*}
  F_1(\alpha) \leq 2 \cdot \left( \frac{d+4}{2} \right)^{d-1} \cdot e^{d-1} =
  2 \left( \frac{e}{2} \cdot (d+4) \right)^{d-1}.  
\end{equation*}

\section{Testing polygonal arcs for the increasing chord property}  \label{sec:alg}

\begin{proof}[Proof of Theorem~\ref{thm:alg}]
  Before presenting our algorithm, we introduce a simpler algorithm that carries out the same task in $O(n^2)$ time.
  
\medskip
\textbf{Algorithm.}
The input is a polygonal chain $P=(p_1, p_2, \dots, p_n)$ in $\RR^d$.
The emptiness tests are with respect to $X$. Execute two loops as follows:

(L1) $X \gets P$. For $i=1,\ldots,n-2$, delete $p_i$ from $X$ and test the emptiness of the negative halfspace
incident to $p_{i+1}$ and orthogonal to the segment $p_i p_{i+1}$ in the chain.

(L2) $X \gets P$. For $i=n,\ldots,3$, delete $p_i$ from $X$ and test the emptiness of the positive halfspace
incident to $p_{i-1}$ and orthogonal to the segment $p_{i-1} p_i$ in the chain.

If all tests return ``empty'', the chain is declared to satisfy the property.

\medskip
The correctness of the algorithm follows from Corollary~\ref{cor:HK}, where we observe that it is enough
to check the condition in Corollary~\ref{cor:HK} only for extremal halfspaces and only at vertices of the chain,
and thus, we need fewer than $2n$ halfspace emptiness tests, each of which can be trivially executed in $O(n)$ time.

Now we show how to obtain a subquadratic time algorithm. To this end, observe that for a fixed point set
with a suitable data structure and preprocessing, a \emph{halfspace emptiness} query can be answered 
in sublinear time (see \cite[Ch.~40]{Aga17} for details).

Again, the input is a polygonal chain $P=(p_1, p_2, \dots, p_n)$ in $\RR^d$. Let $k=\lfloor d/2 \rfloor$. 
Construct an array of data structures with the $n$ points in $P$ for halfspace emptiness
queries~\cite[Ch.~40]{Aga17}. Save these data structures and execute the loops (L1) and (L2) above.
If all tests return ``empty'', the chain is declared to satisfy the property.

Our modification takes into account the fact that
halfspace emptiness queries are \emph{search decomposable} problems,
in the spirit of Bentley~\cite{Be79}. In our description of the modified algorithm, for simplicity,
and without affecting the results, we omit floors and ceilings in specifying certain integer parameters. 
Specifically, for a suitable $q$ to be determined,
we use $q$ data structures for halfspace emptiness queries, each over $n/q$ points,
containing the points in the order they appear in the input chain.
Let these structures be $D_1,D_2,\ldots,D_q$, where  
$D_1$ contains the first $n/q$ points, $D_2$ contains the next $n/q$ points, and so on.

The data structure for $n$ points takes $O(n)$ space $O(n \log{n})$ time to construct,
whereas the expected query time is $O\left(n^{1-1/k}\right) \polylog(n)$; the $\polylog(n)$ factor 
only appears for odd $d$. See the survey by Agarwal~\cite[p.~1069]{Aga17} and the paper by Chan~\cite{Cha12} for details;
note that his algorithm is randomized. 

We next explain how to execute the current step of (L1), with the structures being scanned
in increasing order of their indexes; the process for (L2) is done in reverse order, but otherwise
is completely analogous. 
When executing the $i$th step, the data structure containing $p_i$, say, $D_j$,
may contain some points that have been removed from consideration by the algorithm.
$D_j$ is identified ($j =j(i)$), and subjected to a brute force search against the points that are still present,
and this is followed by a faster search in $D_{j+1},\ldots,D_q$ with the query times from the standard version.
Note that all the points stored within are still present, while all points in the previous
structures $D_1,D_2,\ldots,D_{j-1}$ have been already ``deleted'' by the algorithm,
\ie, the algorithm only uses $D_j,D_{j+1},\ldots,D_q$. The times for the two types of search are
\[ O\left( \frac{n}{q} \right) \text{ and } O\left( q \left( \frac{n}{q} \right)^{1-1/k} \polylog(n/q) \right), \]
respectively. The two terms are (approximately) balanced by setting $q = n^{1/(k+1)}$,
and the expected time for the $i$th step becomes $O\left(n^{1-1/(k+1)} \polylog(n) \right)$. 
Since there are fewer than $2n$ halfspace emptiness queries, the overall expected time is
$O\left(n^{2-1/(k+1)} \polylog(n) \right)$. 
\end{proof}

\section{Concluding remarks} \label{sec:remarks}

Some interesting questions remain:

\begin{enumerate} \itemsep 2pt

\item The upper bound in Theorem~\ref{thm:d-space} on the maximum length of a curve with increasing chords
in $\RR^d$ is surely far from the truth. Can one deduce a bound that is polynomial in $d$?

\item It is conceivable that the semidynamic fixed order deletions executed by the algorithm are amenable
to a dynamization in the style of Bentley \& Saxe~\cite{BS80} or to another speedup technique.
That may lead to a slightly faster algorithm running in $O\left(n^{2-1/k} \polylog(n) \right)$ time,
where $k=\lfloor d/2 \rfloor$. This remains to be confirmed. 

\end{enumerate}

\end{document}